\documentclass[journal]{IEEEtran}

\usepackage{graphicx}  

\usepackage{amsmath}   
\usepackage{amsfonts}

\newcommand{\mean}[1]{\langle{#1}\rangle}

\newcommand{\dgg}{^{\dagger}}
\newcommand{\Tr}{{\rm Tr}\hspace{0.07cm}}

\newtheorem{theorem}{Theorem}[section]
\newtheorem{lemma}{Lemma}[section]
\newtheorem{definition}{Definition}[section]

\newtheorem{corollary}{Corollary}[section]
\newtheorem{proposition}{Proposition}[section]

\begin{document}


\title{
Decoherence-free linear quantum subsystems
}

%
%

\author{Naoki Yamamoto
\thanks{This work was supported in part by JSPS Grant-in-Aid 
No. 40513289. }
\thanks{N. Yamamoto is with the Department of 
        Applied Physics and Physico-Informatics, Keio University, 
        Hiyoshi 3-14-1, Kohoku, Yokohama 223-8522, Japan. 
        (e-mail: yamamoto@appi.keio.ac.jp). }
}
\maketitle





\begin{abstract}

This paper provides a general theory for characterizing and 
constructing a decoherence-free (DF) subsystem for an infinite 
dimensional linear open quantum system. 
The main idea is that, based on the Heisenberg picture of the 
dynamics rather than the commonly-taken Schr\"{o}dinger picture, 
the notions of controllability and observability in control 
theory are employed to characterize a DF subsystem. 
A particularly useful result is a general if and only if 
condition for a linear system to have a DF component; 
this condition is used to demonstrate how to actually construct 
a DF dynamics in some specific examples. 
It is also shown that, as in the finite dimensional case, we 
are able to do coherent manipulation and preservation of a state 
of a DF subsystem. 

\end{abstract}


\section{Introduction}

Quantum information processing shows maximum performance when 
it is run based on an ideal closed system. 
However, in reality any quantum system is an open system 
interacting with surrounding environment and has to be affected 
by {\it decoherence}, which can seriously degrade the information. 
Therefore the theory of {\it decoherence-free subspaces} (DF subspaces) 
or more generally {\it decoherence-free subsystems} (DF subsystems) 
\cite{Kohout2008,Kohout2010,Knill2000,Lidar1998,Lidar2003,Ticozzi2008,
Ticozzi2009} is very promising for realizing various key 
technologies in quantum information science, such as 
quantum computation \cite{BaconPRL2000,KnightPRL2000,KempePRA2001,
Lidar1999,ZanardiPRL1997,Zanardi2003}, 
memory \cite{BaconPRA2006,Wineland2001,SangouardPRA2006}, 
communication \cite{TeichPRL2003,XuePLA2008}, and 
metrology \cite{DornerNJP2012}. 
It also should be mentioned that some experimental demonstrations 
have been reported \cite{Viola2003,ImotoPRL2011,Viola2001,Pan2006}.

Let us now review the basic idea of DF (or noiseless) subsystems. 
Quantum information processing is founded on the unitarity of 
the ideal closed dynamics of a system of interest \cite{NielsenBook}. 
More precisely, the ``state" $\hat\rho_{\rm cl}$ of a closed quantum 
system is driven by an Hermitian operator (Hamiltonian) $\hat H$ and 
is subjected to the dynamics 
$d\hat{\rho}_{\rm cl}/dt = -i[\hat H, \hat\rho_{\rm cl}] = 
-i(\hat H \hat\rho_{\rm cl} - \hat\rho_{\rm cl} \hat H)$, which yields 
the solution 
$\hat\rho_{\rm cl}(t)
={\rm e}^{-i\hat H t}\hat\rho_{\rm cl}(0){\rm e}^{i\hat H t}$ 
(a detailed description is given in Appendix); 
this unitary change of $\hat\rho_{\rm cl}(0)$ represents the quantum 
information processing under consideration, in which case it is often called 
the {\it gate operation} on $\hat\rho_{\rm cl}(0)$. 
However, in reality any quantum system interacts with environment 
and does not obey a unitary dynamics. 
Especially when the interaction is instantaneous and a Markovian 
approximation can be taken, the system's state evolves 
in time with the following {\it master equation}: 
\begin{equation}
\label{general master equation}
   \frac{d\hat \rho}{dt} 
    = - i[\hat H , \hat\rho] 
      + \sum_{i=1}^m \Big( 
          \hat L_i \hat\rho \hat L_i^* 
           -\frac{1}{2} \hat L_i^* \hat L_i \hat\rho 
            -\frac{1}{2} \hat\rho \hat L_i^* \hat L_i \Big).
\end{equation}
The operator $\hat L_i$ represents the coupling between the system and 
the $i$-th environment field; this second term represents the decoherence. 
Due to this, the time-evolution of $\hat\rho$ is not unitary, implying that 
in general coherent manipulation or preservation of an arbitrary state of 
an open system is impossible. 
Nevertheless, it is known in the {\it finite dimensional} case that we can 
have a chance to protect a specific part of the state against the environment. 
That is, if the system matrices $\hat H$ and $\hat L_k$ 
satisfy specific algebraic conditions and an appropriate initial state 
$\hat\rho(0)$ is prepared, then in certain basis vectors the system 
state has a block diagonal form 
$\hat\rho(t)
={\rm diag}\{\hat\rho_{\rm DF}(t)\otimes \hat\rho_{\rm D}(t), O\}$ 
for all $t\geq 0$, with $\otimes$ denoting the Kronecker product; 
$\hat\rho_{\rm DF}$ is the state of the DF subsystem, which obeys 
a unitary time evolution, while the decohered (D) component 
$\hat\rho_{\rm D}$ is still subjected to a non-unitary dynamics. 
That is, the dynamics of these states are respectively given by 
\begin{eqnarray}
& & \hspace*{-2em}
\label{DF dynamics in introduction}
    \frac{d\hat \rho_{\rm DF}}{dt} 
    = - i[\hat H_{\rm DF}, \hat\rho_{\rm DF}], 
\\ & & \hspace*{-2em}
\label{D dynamics in introduction}
    \frac{d\hat \rho_{\rm D}}{dt} 
    = - i[\hat H' , \hat\rho_{\rm D}] 
\nonumber \\ & & \hspace*{0em}
    \mbox{}
     + \sum_{i=1}^m \Big( 
      \hat L'_i \hat\rho_{\rm D} \hat L_i'\mbox{}^* 
       -\frac{1}{2} \hat L_i'\mbox{}^* \hat L'_i \hat\rho_{\rm D} 
        -\frac{1}{2} \hat\rho_{\rm D} \hat L_i'\mbox{}^* \hat L'_i \Big). 
\end{eqnarray}
When $\dim\hat \rho_{\rm D}=1$ (i.e. $\hat \rho_{\rm D}=1$), 
the space of $\hat \rho_{\rm DF}$ is called the DF subspace. 
Note that the complement ${\rm diag}\{O, \hat\rho_{{\rm D}'}(t)\}$ is also 
decohered and its dynamics is governed by a master equation. 
Equation~\eqref{DF dynamics in introduction} means that, by 
appropriately devising the reduced Hamiltonian $\hat H_{\rm DF}$, 
we can carry out a desirable unitary gate operation.

Now we come up with the following question; 
is it possible to construct a DF subsystem in the {\it infinite 
dimensional} case? 
This question is equivalent to asking the applicability of the 
above theory for discrete variable (DV) systems to continuous 
variable (CV) systems \cite{BraunsteinRMP,Furusawa2011}. 
Indeed, despite the recent rapid progress of quantum information 
science with infinite-dimensional/CV systems, there has not been 
developed a general theory of DF subsystems; 
it was only shown in \cite{Beny2009} that the infinite dimensional 
DF subsystems can be characterized in terms of operator algebras, 
which is though hard to use for engineering purpose. 
The problem becomes relatively easy if we focus on the case where 
there is no decohered subspace component, i.e. the case where 
the decomposition 
$\hat\rho(t) = \hat\rho_{\rm DF}(t)\otimes \hat\rho_{\rm D}(t)$ is 
possible. 
Actually in this case some infinite dimensional systems having a 
DF subsystem are known, in which case the DF component is often 
called the {\it dark mode}; 
the DF mode is utilized for state generation and preservation 
\cite{Huang2013,Zambrini2013,PrauznerJPA2004}, cooling of a mechanical 
oscillator \cite{Dong2012}, and state transfer \cite{Clerk2012}. 
However there is no general theory unifying the methods taken in 
these examples. 
Particularly what is lacking is the general procedure for constructing 
a DF subsystem, which has actually been found in the finite dimensional 
case \cite{Choi2006}.

This paper provides a basic framework dealing with a general 
{\it linear quantum system} containing a DF subsystem, when 
especially there is no decohered subspace component. 
A linear system is an infinite dimensional open system that 
describes the dynamics of, for instance, optical modes of a 
light field \cite{GardinerBook,IidaTAC2012,WallsBook,WisemanBook,
Yanagisawa-2003a,Yanagisawa-2003b}, position of 
an opto-mechanical oscillator \cite{LawPRA1995,YamamotoPRA2012}, 
vibration mode of a trapped particle \cite{WinelandRMP2003}, 
collective spin component of a large atomic ensemble 
\cite{PolzikRMP2010,PolzikNaturePhys2011,Muschik2011}, or more 
generally canonical conjugate pairs of a harmonic oscillator 
network \cite{BurgarthPRL2012,GoughPRA2008,GoughPRA2010,
JamesTAC2008,NurdinAutomatica2009,NurdinSIAM2009,
WisemanPRL2005,YamamotoPRA2006}, which all serve as quantum 
information devices.

The key ideas are twofold. 
The first main idea is that we focus on the dynamics of {\it finite} 
number of physical quantities (called ``observables" in physics) 
that completely characterize the system of interest, rather than 
the state that must be defined on the whole infinite dimensional 
space. 
This means that we represent the system in the {\it Heisenberg 
picture}, rather than the Schr\"{o}dinger picture describing the 
dynamics of a state. 
Then, the system dynamics is not described by the master equation 
\eqref{general master equation}, but by a {\it quantum stochastic 
differential equation (QSDE)} 
\cite{Belavkin1992,Bouten2007,GardinerBook,Hudson1984,WisemanBook}. 
The QSDE explicitly represents the system-environment coupling 
in a form that the equation contains the terms of environmental input 
and output fields. 
Together with the fact that we are now dealing with the finite 
number of observables, this critical feature of the QSDE further 
leads us to have the second main idea. 
That is, we are now able to naturally define a DF observable as 
follows; 
i.e., it is an observable such that its dynamics is not affected by the 
environmental input field and the behavior of that observable does not 
appear in the output field. 
Actually it will be proven in this paper that this definition is equivalent to 
the original one, for the finite dimensional case. 
The significant point of this definition is that a DF subsystem 
is now clearly characterized in terms of the notions of 
{\it controllability} and {\it observability}, which are well 
developed in control theory (e.g., \cite{AstromBook,KailathBook}). 
In particular, in the linear case, such controllability and 
observability properties can be straightforwardly captured by 
simple matrix algebras that only require us to compute the rank 
of certain matrices. 
As a result, based on the QSDE representation of a given linear 
open quantum system, we obtain a necessary and sufficient condition 
for that system to have a DF subsystem, which further leads to 
an explicit procedure for constructing the DF subsystem. 
This result will be given in Section~III.

Once a DF subsystem is constructed, then we should be interested 
in the {\it stability} of the dynamics. 
In the current context, this is the property that autonomously 
erases the correlation between the DF and the D parts, and it 
has been extensively investigated for finite dimensional systems 
\cite{SchirmerPRA2010,Ticozzi2008,Ticozzi2009}. 
In Section~IV, particularly in the case of Gaussian systems 
\cite{Ferraro2005,WeedbrookRMP2012}, a simple if and only if 
criterion that guarantees the stability of a DF subsystem will be 
provided.

The remaining part of the paper is devoted to the study 
of some concrete examples of a DF linear subsystem, which is 
divided into two topics as follows.

{\bf (i)} (Section V) 
The first topic deals with coherent manipulation and preservation 
of a state of a DF subsystem. 
More specifically, it is demonstrated that, in a specific two-mode 
linear system having a one-mode DF component, arbitrary linear 
unitary gate operation and state preservation are possible under 
relatively mild assumptions. 
Also, for the same system, it will be shown that preserving a state 
of the DF subsystem is equivalent to preserving a certain {\it entangled 
state} among the whole two-mode system. 
In addition, the stability property is examined, showing its importance 
in realizing robust state manipulation and preservation with the DF 
subsystem.

{\bf (ii)} (Section VI) 
The second is about how to engineer a DF subsystem for a given 
open linear system. 
We actually find in two specific examples that, by appropriately 
devising an auxiliary system coupled to the original one, it is 
possible to construct a physically meaningful DF subsystem. 
Therefore, in addition to the practical merit that a DF subsystem 
can be used for coherent quantum information processing, a DF 
subsystem is useful in the sense that it simulates an ideal closed 
quantum system even in a realistic open environment.

We use the following notations: 
for a matrix $A=(a_{ij})$, the symbols $A^\dagger$, $A^\top$, 
and $A^\sharp$ represent its Hermitian conjugate, transpose, and 
elementwise complex conjugate of $A$, i.e., 
$A^\dagger=(a_{ji}^*)$, $A^\top=(a_{ji})$, and 
$A^\sharp=(a_{ij}^*)=(A\dgg)^\top$, respectively. 
For a matrix of operators, $\hat A=(\hat a_{ij})$, we use the 
same notation, in which case $\hat a_{ij}^*$ denotes the adjoint 
to $\hat a_{ij}$. 
$I_n$ denotes the $n\times n$ identity matrix. 
$\Re$ and $\Im$ denote the real and imaginary parts, respectively. 
$O$ is a zero matrix with appropriate dimension. 
$\otimes$ is the tensor product, which is the Kronecker product 
in the finite dimensional case. 
${\rm Ker}(A)=\{x \hspace{0.05cm}|\hspace{0.05cm} Ax=0\}$ and 
${\rm Range}(A)=\{y \hspace{0.05cm}|\hspace{0.05cm} y=Ax,~\forall x\}$ 
denote the kernel and the range of a matrix $A$, respectively. 
Throughout the paper, we set $\hbar=1$.

Some fundamentals of quantum mechanics, e.g., states and 
observables, are given in Appendix.


\section{Preliminaries}

\subsection{The QSDE}

Let us consider an open system interacting with some environment 
fields, which are assumed to be independent vacuum fields for 
simplicity. 
The time-evolution of an observable of this system is described 
by a QSDE, as shown below. 
For a comprehensive introduction to the theory of quantum 
stochastic calculus, we refer to \cite{Bouten2007}.

Let $\hat a_i(t)$ be the annihilation operator of the $i$-th 
vacuum field and assume that $\hat a_i(t)$ instantaneously 
interacts with the system, which means that it satisfies the 
{\it canonical commutation relation} (CCR) 
$[\hat a_i(s), \hat a_j^*(t)]=\delta_{ij}\delta(t-s)$. 
This relation reminds us the classical white noise; 
actually the field annihilation process 
$\hat A_i(t)=\int_0^t \hat a_i(s)ds$, which is the quantum version 
to the classical Wiener process, satisfies the following 
{\it quantum Ito rule} (in the vacuum field) 
\cite{Belavkin1992,Bouten2007,GardinerBook,Hudson1984,WisemanBook}: 
\[
    d\hat A_i d\hat A_j^*=\delta_{ij}dt,~~
    d\hat A_i d\hat A_j=d\hat A_i^* d\hat A_j^*
                       =d\hat A_i^* d\hat A_j=0, 
\]
where the adjoint operator $\hat A_i^*(t)$ represents the field creation 
process. 
The system-field interaction in the time interval $[t, t+dt)$ is 
described by the Hamiltonian 
$\hat H_{\rm int}(t+dt, t)
=i\sum_i(\hat L_i d\hat A_i^*(t) - \hat L_i^* d\hat A_i(t))$, 
where $\hat L_i$ is a system operator representing the coupling 
with the $i$-th vacuum field. 
A system observable $\hat X(0)$ evolves in the Heisenberg 
picture to $\hat X(t)=\hat U^*(t)\hat X(0) \hat U(t)$, where 
$\hat U(t)$ is the unitary operator from time $0$ to $t$, which 
is constructed from the relation $\hat U(t+dt)=\hat U(t+dt,t)\hat U(t)$ 
with $\hat U(t+dt, t)={\rm exp}[-i\hat H_{\rm int}(t+dt, t)]$. 
Then, using the above quantum Ito rule we can derive the QSDE of 
$\hat X(t)$: 
\begin{eqnarray}
& & \hspace*{-2.5em}
\label{QSDE}
    d\hat X(t)=\Big( i[\hat H(t), \hat X(t)] 
\nonumber \\ & & \hspace*{-1em}
    \mbox{}
       + \sum_{i=1}^m \big( 
         \hat L_i^*(t) \hat X(t) \hat L_i(t)
          -\frac{1}{2} \hat L_i^*(t) \hat L_i(t) \hat X(t)
\nonumber \\ & & \hspace*{8.7em}
    \mbox{}
           -\frac{1}{2} \hat X(t) \hat L_i^*(t) \hat L_i(t) \big) \Big)dt
\nonumber \\ & & \hspace*{-1em}
    \mbox{}
      + \sum_{i=1}^m \Big( 
          [\hat X(t), \hat L_i(t)]d\hat A_i^*(t)
           -  [\hat X(t), \hat L_i^*(t)]d\hat A_i(t) \Big). 
\end{eqnarray}
Here, we have added a system Hamiltonian $\hat H$ and defined 
$\hat H(t)=\hat U^*(t)\hat H\hat U(t)$ and 
$\hat L_i(t)=\hat U^*(t)\hat L_i\hat U(t)$. 
The master equation \eqref{general master equation} of the system's 
unconditional state, $\hat\rho(t)$, is obtained through the time evolution 
of the mean value of $\hat X(t)$ and the relation 
$\mean{\hat X(t)}=\Tr[\hat X(0)\hat \rho(t)]$. 
The change of the field operator is also obtained; 
the output field 
$\hat A^{\rm out}_i(t)=\hat U^*(t) \hat A_i(t) \hat U(t)$ 
after the interaction satisfies 
\begin{equation}
\label{general output}
    d\hat A^{\rm out}_i(t) = \hat L_i(t)dt + d\hat A_i(t). 
\end{equation}
%


\subsection{Quantum linear systems}

In this paper, we consider a system composed of $n$ subsystems. 
The variable of each subsystem is called the {\it mode}; 
particularly in our case the $i$-th mode is specified by the canonical 
conjugate pairs $(\hat q_i, \hat p_i)$ satisfying the CCR 
$[\hat q_i, \hat p_j]=i \delta _{ij}$. 
Note that both $\hat q_i$ and $\hat p_i$ are infinite dimensional 
operators, hence the system is of infinite dimensional. 
Let us here define the vector of operators 
$\hat x =(\hat q_1, \hat p_1, \ldots, \hat q_n, \hat p_n)^\top$. 
Then the CCRs $[\hat q_i, \hat p_j]=i \delta _{ij}$ are summarized as 
\begin{equation}
\label{CCR}
   \hat x \hat x ^\top -(\hat x \hat x^\top )^\top 
      = i \Sigma_n,
\end{equation}
where
\[
   \Sigma_n = {\rm diag}\{\Sigma, \ldots, \Sigma \},~~
   \Sigma = \left(\begin{array}{cc}
               0 & 1 \\
               -1 & 0
            \end{array}\right). 
\]
($\Sigma_n$ is a $2n\times 2n$ block diagonal matrix.) 
We are interested in a system whose Hamiltonian and coupling 
operator are respectively given by 
\[
     \hat H = \hat x^\top G \hat x/2,~~
     \hat L_i = c_i^\top \hat x, 
\]
where $G=G^\top \in {\mathbb R}^{2n \times 2n}$ and 
$c_i \in {\mathbb C}^{2n}~(i=1,\ldots,m)$ 
\cite{WisemanPRL2005,WisemanBook,YamamotoPRA2006}. 
Also let us define the canonical conjugate pairs of the noise process: 
\begin{equation}
\label{noise quadrature}
   \hat Q_i=(\hat A_i + \hat A_i^*)/\sqrt{2},~~~
   \hat P_i=(\hat A_i - \hat A_i^*)/\sqrt{2}i, 
\end{equation}
and collect them into a single vector of operators as 
$\hat {\cal W}=(\hat Q_1, \hat P_1, \dots, \hat Q_m, \hat P_m)^\top$. 
Then, from the QSDE \eqref{QSDE}, the vector of system variables 
$\hat x(t)=(\hat q_1(t), \hat p_1(t), \ldots, 
\hat q_n(t), \hat p_n(t) )^\top$, 
where $\hat q_i(t)=\hat U^*(t)\hat q_i\hat U(t)$ and 
$\hat p_i(t)=\hat U^*(t)\hat p_i\hat U(t)$, satisfies the following 
linear equation:
\begin{equation}
\label{linear QSDE}
   d\hat x(t) = A\hat x(t) dt 
       + \Sigma_n C^\top \Sigma_m d\hat {\cal W}(t).
\end{equation}
The coefficient matrices are given by 
\begin{eqnarray}
& & \hspace*{-2em}
      A=\Sigma_n(G+C^\top\Sigma_m C/2)
                   \in {\mathbb R}^{2n\times 2n},
\nonumber \\ & & \hspace*{-2em}
      C = \sqrt{2}(\Re(c_1), \Im(c_1), \ldots, \Re(c_m), \Im(c_m))^\top
       \in {\mathbb R}^{2m\times 2n}.
\nonumber
\end{eqnarray}
This specific structure of the system matrices is due to the 
unitary evolution of $\hat q_i(t)$ and $\hat p_i(t)$, which is 
indeed necessary to satisfy the CCR 
$[\hat q_i(t), \hat p_j(t)]=i \delta _{ij}$ for all $t$. 
Moreover, corresponding to the canonical conjugate representation of the 
input field \eqref{noise quadrature}, let us define the output 
process as 
\[
   \hat Q_i^{\rm out}
     =(\hat A_i^{\rm out} + \hat A_i^{\rm out}\mbox{}^*)/\sqrt{2},~~
   \hat P_i^{\rm out}
     =(\hat A_i^{\rm out} - \hat A_i^{\rm out}\mbox{}^*)/\sqrt{2}i, 
\]
and collect them as 
$\hat{\cal W}^{\rm out}
=(\hat Q_1^{\rm out}, \hat P_1^{\rm out},\ldots,
\hat Q_m^{\rm out}, \hat P_m^{\rm out})^\top$. 
Then, from \eqref{general output} we have 
\begin{equation}
\label{linear output}
   d\hat {\cal W}^{\rm out}(t) 
     = C\hat x(t) dt + d\hat {\cal W}(t). 
\end{equation}
To denote the system variable, in what follows we often omit the time 
index $t$ and simply use $\hat x$; 
to avoid confusion, the initial value is explicitly written as 
$\hat x(0)$.

Here we remark that the linearly transformed variable 
$\hat x'=T^\top \hat x$ must satisfy the CCR \eqref{CCR}, which 
implies that 
$x'x'\mbox{}^\top-(x'x'\mbox{}^\top)^\top
=T^\top(xx^\top-(xx^\top)^\top)T=iT^\top\Sigma_n T=i\Sigma_n$. 
Consequently, in the quantum case, the similarity transformation $T$ 
must be {\it symplectic}, meaning that it satisfies 
\begin{equation}
\label{symplectic}
    T^\top \Sigma_n T = \Sigma_n. 
\end{equation}
Note further that, in the linear case, the time evolution is also 
a symplectic transformation, because we can find a symplectic 
matrix $S$ satisfying $\hat x(t)=S^\top \hat x(0)$.

%
%

Now let $\mean{\hat x}$ be the mean vector, where the mean operation 
$\mean{\cdot}$ is taken elementwise. 
The covariance matrix is defined by 
\[
    V=\mean{\Delta \hat x \Delta \hat x^\top 
              + (\Delta \hat x \Delta \hat x^\top)^\top}/2,~~
    \Delta \hat x=\hat x-\mean{\hat x}.
\]
Note that the uncertainty relation $V+i\Sigma_n/2 \geq 0$ holds. 
Then for the linear system \eqref{linear QSDE}, the time-evolutions of 
$\mean{\hat x}$ and $V$ are respectively given by 
\begin{equation}
\label{mean and variance dynamics}
    d\mean{\hat x}/dt = A \mean{\hat x},~~~
    dV/dt = AV +V A^\top + D,
\end{equation}
where $D=\Sigma_n^\top C^\top C\Sigma_n/2$.


\subsection{Gaussian systems}
\label{Preliminary: Gaussian states}

A quantum Gaussian state can be characterized by only the mean 
vector $\mean{\hat x}$ and the covariance matrix $V$ 
\cite{Ferraro2005,WeedbrookRMP2012}. 
More specifically, its Wigner function is identical to the Gaussian 
probabilistic distribution with mean $\mean{\hat x}$ and 
covariance $V$. 
A notable property of a linear system is that it preserves the 
Gaussianity of the state; 
that is, if the initial state of the system \eqref{linear QSDE} is 
Gaussian, then, for all later time $t$ the state is Gaussian with 
mean $\mean{\hat x(t)}$ and covariance $V(t)$, whose time evolutions 
are given by the linear differential equations 
\eqref{mean and variance dynamics}. 
A unique steady state of this Gaussian system exists only when $A$ 
is a {\it Hurwitz} matrix, i.e., all the eigenvalues of $A$ have 
negative real parts. 
If it exists, the mean vector is $\mean{\hat x(\infty)}=0$, and 
the covariance matrix $V(\infty)$ is given by the unique 
solution to the algebraic Lyapunov equation 
$AV(\infty) + V(\infty) A^\top + D = O$. 
Note that several non-Gaussian states can be generated in a linear quantum 
system \cite{GuofengTAC2012}.


\subsection{Controllability and observability}
\label{controllability and observability}

Here we review the notions of controllability and observability, 
which play the fundamental roles in characterizing various important 
properties of a system, such as stabilizability of the system via 
control input or capability of constructing an estimator that 
continuously monitors the system variables. 
See e.g. \cite{AstromBook,KailathBook} for full description 
of those theories.

Let us consider the following general linear system: 
\begin{equation}
\label{linear classical system}
    \dot{x}=Ax+Bu,~~~~y=Cx,
\end{equation}
where $x$ is the system variable of dimension $n$, and $u$ and $y$ 
are an input and an output with certain dimensions, respectively. 
Note that $u$ represents any input such as a disturbing noise and 
a tunable control signal. 
Now define the {\it controllability matrix} 
${\cal C}=(B, AB, \ldots, A^{n-1}B)$ and assume that 
${\rm rank}({\cal C})=n'$. 
(Note $n'\leq n$ in general.) 
Then, there exists a linear coordinate transformation 
$x\rightarrow x'=(x_1'\mbox{}^\top, x_2'\mbox{}^\top)^\top$ such 
that the system \eqref{linear classical system} is transformed to
\begin{eqnarray}
& & \hspace*{-2em}
     \frac{d}{dt}
      \left( \begin{array}{c} 
                 x_1' \\
                 x_2' \\
               \end{array}\right)
     =\left( \begin{array}{cc} 
                 A_{11}' & O \\
                 A_{21}' & A_{22}' \\
               \end{array}\right)
       \left( \begin{array}{c} 
                 x_1' \\
                 x_2' \\
               \end{array}\right)
        +\left( \begin{array}{c} 
                 O \\
                 B_2' \\
               \end{array}\right)u,
\nonumber \\ & & \hspace*{2.2em}
      y=(C_1',~C_2')
       \left( \begin{array}{c} 
                 x_1' \\
                 x_2' \\
               \end{array}\right), 
\nonumber
\end{eqnarray}
where $x_1'$ is of dimension $n-n'$. 
This equation shows that $x_1'$ is not affected by $u$, 
hence it is called {\it uncontrollable} with respect to (w.r.t.) $u$. 
Of course, if $n'=n$, or equivalently if ${\cal C}$ is of full 
rank, then all the elements of $x$ is affected by $u$ and in this case 
the system is called controllable. 
Next, let us define the {\it observability matrix} 
${\cal O}=(C^\top, A^\top C^\top, \ldots, (A^\top)^{n-1}C^\top)^\top$ 
and assume that ${\rm rank}({\cal O})=n''$. 
Then, there exists a linear coordinate transformation that transforms 
\eqref{linear classical system} to the system of the following form: 
\begin{eqnarray}
& & \hspace*{-2em}
     \frac{d}{dt}
      \left( \begin{array}{c} 
                 x_1'' \\
                 x_2'' \\
               \end{array}\right)
     =\left( \begin{array}{cc} 
                 A_{11}'' & A_{12}'' \\
                 O & A_{22}'' \\
               \end{array}\right)
       \left( \begin{array}{c} 
                 x_1'' \\
                 x_2'' \\
               \end{array}\right)
        +\left( \begin{array}{c} 
                 B_1'' \\
                 B_2'' \\
               \end{array}\right)u,
\nonumber \\ & & \hspace*{2.3em}
      y=(O,~C_2'')
       \left( \begin{array}{c} 
                 x_1'' \\
                 x_2'' \\
               \end{array}\right), 
\nonumber
\end{eqnarray}
where $x_1''$ is of dimension $n-n''$. 
Clearly, $x_1''$ cannot be seen from the output $y$, hence $x_1''$ 
is called {\it unobservable} w.r.t. $y$. 
Note here that the term ``unobservable" does not mean that $x_1''$ 
is not a physical quantity.


\section{Characterization of linear DF subsystems}

In this section, based on the QSDE representation of the system, we 
show the if and only if condition for a general linear quantum system 
to have a DF subsystem, particularly in the case where there is no 
decohered subspace component. 
Moreover, we will see that this condition leads to a concrete procedure for 
constructing the DF subsystem.

Here we remark that in \cite{Beny2007a,Beny2007b} the Heisenberg 
picture is taken to investigate the quantum error correction, in the 
general finite dimensional case. 
This approach differs from ours in that the authors have examined 
the {\it mean} dynamics of the system observables with its 
environment fields traced out (averaged); 
that is, unlike the QSDE, the time-evolution equation does not contain 
the input and output terms, so the notions of controllability and 
observability are not applicable.


\subsection{The Heisenberg picture description of a DF subsystem}

In this subsection, we describe a finite dimensional open system 
having a DF subsystem, in the Heisenberg picture. 
As mentioned in Section~I, if the system contains a DF subsystem, 
then its state, which is subjected to the master equation 
\eqref{general master equation}, is represented in a specific 
basis as 
$\hat\rho(t)
={\rm diag}\{\hat\rho_{\rm DF}(t)\otimes \hat\rho_{\rm D}(t), O\}$ with 
$\hat \rho_{\rm DF}$ governed by the unitary dynamics 
\eqref{DF dynamics in introduction}. 
Ticozzi and Viola \cite{Ticozzi2008,Ticozzi2009} gave an explicit iff condition 
for the system to have a DF subsystem; 
in particular when there is no decohered subspace, the condition is that 
the system matrices are of the following form: 
\begin{equation}
\label{Ticozzi condition}
    \hat H = \hat H_{\rm DF}\otimes I_{\rm D} 
                + I_{\rm DF}\otimes \hat H_{\rm D},~~~
    \hat L_i = I_{\rm DF}\otimes \hat L_{{\rm D},i}. 
\end{equation}
Actually, by substituting these matrices with initial state 
$\hat\rho(0)=\hat\rho_{\rm DF}(0)\otimes \hat\rho_{\rm D}(0)$ for the 
master equation \eqref{general master equation}, we obtain 
\eqref{DF dynamics in introduction} and \eqref{D dynamics in introduction}.

Let us now change the picture to the Heisenberg's one and look 
at the dynamics of an open system, the QSDE \eqref{QSDE}. 
Now we assume that the system matrices are given by \eqref{Ticozzi condition}; 
then the system observables 
$\hat X_{\rm DF}(t)
=\hat U^*(t)(\hat X_{\rm DF}(0)\otimes I_{\rm D}) \hat U(t)$ and 
$\hat X_{\rm D}(t)
=\hat U^*(t)(I_{\rm DF}\otimes\hat X_{\rm D}(0)) \hat U(t)$ satisfy 
the following dynamical equations:
\begin{eqnarray}
& & \hspace*{-2em}
\label{QSDE - DF}
    d\hat X_{\rm DF}=i[\hat H_{\rm DF}, \hat X_{\rm DF}]dt, 
\\ & & \hspace*{-2em}
\label{QSDE - D}
    d\hat X_{\rm D}=\Big( i[\hat H_{\rm D}, \hat X_{\rm D}] 
\nonumber \\ & & \hspace*{1em}
    \mbox{}
     + \sum_{i=1}^m \big( 
     \hat L_{{\rm D},i}^* \hat X_{\rm D} \hat L_{{\rm D},i}
     -\frac{1}{2} \hat L_{{\rm D},i}^* \hat L_{{\rm D},i} \hat X_{\rm D}
\nonumber \\ & & \hspace*{9.4em}
    \mbox{}
     -\frac{1}{2} \hat X_{\rm D} \hat L_{{\rm D},i}^* \hat L_{{\rm D},i} 
       \big) \Big)dt
\nonumber \\ & & \hspace*{1em}
    \mbox{}
      + \sum_{i=1}^m \Big( 
          [\hat X_{\rm D}, \hat L_{{\rm D},i}]d\hat A_i^*
           -  [\hat X_{\rm D}, \hat L_{{\rm D},i}^*]d\hat A_i \Big),
\end{eqnarray}
where we have omitted the time index $(t)$. 
Also, the output equation \eqref{general output} becomes 
\begin{equation}
\label{general output DF D}
    d\hat A^{\rm out}_i(t) = \hat L_{{\rm D},i}(t)dt + d\hat A_i(t). 
\end{equation}
Note that $\hat L_{{\rm D},i}(t)$ evolves in time according to \eqref{QSDE - D}. 
Thus, the above equations imply that $\hat X_{\rm DF}(t)$ is completely isolated 
from the environment. 
In other words, $\hat X_{\rm DF}(t)$ is not affected by the input field 
$\hat{A}_i(t)$, and the output field $\hat A^{\rm out}_i(t)$ does not 
contain any information about that observable for all $i$. 
In terms of control theory, therefore, an observable that is decoherence 
free is uncontrollable w.r.t. $\hat{A}_i(t)$ and unobservable w.r.t. 
$\hat{A}_i^{\rm out}(t)$ for all $i$.


\subsection{The DF condition for linear systems}

In the above subsection we have seen that a DF observable is 
uncontrollable and unobservable. 
In the finite dimensional case, this is indeed true for {\it any} 
DF observable, because \eqref{Ticozzi condition} is necessary and 
sufficient. 
This characterization is reasonable from both physics and engineering 
viewpoints, and thus, it may be taken as a definition for a system 
to have a DF subsystem even in the infinite dimensional case. 
In particular, for the linear quantum system \eqref{linear QSDE} and 
\eqref{linear output}, we pose the following formal statement:

\begin{definition}
\label{Def of DF mode}
For the linear quantum system \eqref{linear QSDE} and \eqref{linear output}, 
suppose that (i) there exists a subsystem that is uncontrollable w.r.t. $\hat{\cal W}$ 
and unobservable w.r.t. $\hat{\cal W}^{\rm out}$, and further, (ii) its variable 
satisfies the CCR \eqref{CCR}. 
Then, that subsystem is called the DF subsystem, and its variable $\hat x_{\rm DF}$ 
is called the DF mode. 
\end{definition}

The requirement for $\hat x_{\rm DF}$ to satisfy the CCR property stems 
from the fact that, in CV quantum information processing, any gate 
operation acts on the pairs of canonical conjugate observables, 
and this discerns our problem from a simple one where we are interested 
in a merely uncontrollable and unobservable subsystem. 
A different view of this requirement is that it corresponds to the 
{\it physical realizability} condition \cite{JamesTAC2008} of the DF 
subsystem; 
indeed, if the subsystem is completely isolated, it must be a physical system. 
See \cite{Petersen2013} where, in a specific optical system, an uncontrollable 
and unobservable subsystem is constructed using the physical realizability 
condition.

The purpose here is to fully characterize the linear system 
having a DF subsystem. 
As in the classical case, the following controllability matrix 
${\cal C}$ and the observability matrix ${\cal O}$ will play a 
key role in deriving such a characterization; 
\begin{eqnarray*}
& & \hspace*{-1em}
    {\cal C}=(\Sigma_n C^\top\Sigma_m, 
              A\Sigma_n C^\top\Sigma_m, \ldots, 
              A^{2n-1}\Sigma_n C^\top\Sigma_m), 
\\ & & \hspace*{-1em}
    {\cal O}=(C^\top, A^\top C^\top, \ldots, 
              (A^\top)^{2n-1}C^\top)^\top. 
\end{eqnarray*}
The following lemma will be useful:
\begin{lemma}
\label{DF lemma}
\[
     \Sigma_n v\in {\rm Ker}({\cal O})~~\Leftrightarrow~~
     v\in {\rm Ker}({\cal C}^\top)
\]
\end{lemma}

\begin{proof} 
Suppose $\Sigma_n v\in {\rm Ker}({\cal O})$. 
This means that we have $CA^k \Sigma_n v=0,~\forall k\geq 0$, which 
is further equivalent to $C(\Sigma_n G)^k \Sigma_n v=0,~\forall k\geq 0$, 
due to the specific structure of $A$, i.e. 
$A=\Sigma_n(G+C^\top\Sigma_m C/2)$. 
This condition yields 
\begin{eqnarray*}
& & \hspace*{-1em}
    (A^k\Sigma_nC^\top\Sigma_m)^\top v
      =\Sigma_m C\Sigma_n(G\Sigma_n^\top
                  +C^\top\Sigma_m C\Sigma_n/2)^k v
\\ & & \hspace*{6.7em}
      =\Sigma_m C\Sigma_n(G\Sigma_n^\top)^k v=0,
\end{eqnarray*}
for all $k\geq 0$, hence $v\in {\rm Ker}({\cal C}^\top)$. 
The reverse direction is readily obtained. 
\end{proof}

This strong relationship between the controllability and observability 
properties follows from the specific structure of the system matrices 
mentioned in Section II-B. 
Interestingly, this relationship can be connected to the CCR as follows. 
Let us define two observables $\hat r_1=v^\top\Sigma_n^\top \hat x$ and 
$\hat r_2=v^\top \hat x$ with $v$ a normalized vector in 
${\rm Ker}({\cal C}^\top)$. 
Then from \eqref{CCR} we find 
\begin{eqnarray*}
& & \hspace*{-1em}
    [\hat r_1,~\hat r_2]
     = \hat r_1 \hat r_2 - \hat r_2 \hat r_1
     = \hat r_1 \hat r_2 - (\hat r_2 \hat r_1)^\top
\\ & & \hspace*{2.4em}
     = v^\top\Sigma_n^\top 
         [ \hat x \hat x^\top - (\hat x \hat x^\top)^\top ] v
     = v^\top\Sigma_n^\top (i\Sigma_n)v
     = i. 
\end{eqnarray*}
This means that, if an observable is free from the input field, 
then its canonical conjugate must not appear in the output field. 
Thus, in the linear case, an uncontrollable physical quantity and 
an unobservable physical quantity are in the relationship of 
canonical conjugation.

The following theorem is our first main result:

\begin{theorem}
\label{DF mode theorem}
The linear system \eqref{linear QSDE} and \eqref{linear output} 
has a DF subsystem if and only if 
\begin{equation}
\label{DF theorem condition}
    {\rm Ker}({\cal O}) \cap {\rm Ker}({\cal O}\Sigma_n) 
       \neq \emptyset.
\end{equation}
Then there always exists a matrix $T_1$ satisfying 
${\rm Range}(T_1)={\rm Ker}({\cal O}) \cap {\rm Ker}({\cal O}\Sigma_n)$ 
and $T_1^\top\Sigma_n T_1=\Sigma_\ell$, and then the DF mode is given by 
$\hat x_{\rm DF}=T_1^\top \hat x$. 
\end{theorem}

\begin{proof} 
First recall that the condition (i) in Def.~\ref{Def of DF mode} means 
${\rm Range}({\cal C})^{\rm c}\cap {\rm Ker}({\cal O}) \neq \emptyset$, 
where the subscript ${\rm c}$ denotes the complement of the set. 
This condition can be equivalently represented by 
${\rm Ker}({\cal C}^\top) \cap {\rm Ker}({\cal O}) \neq \emptyset$, 
because the controllability and obervability properties are invariant 
under the similarity transformation of the basis vectors. 
Now, Lemma~\ref{DF lemma} states that {\it any} vector satisfying 
${\cal O}\Sigma_n v=0$ fulfills ${\cal C}^\top v=0$, and vice versa. 
This means ${\rm Ker}({\cal O}\Sigma_n) = {\rm Ker}({\cal C}^\top)$, 
hence, the condition (i) in Def.~\ref{Def of DF mode} is equivalent to 
\eqref{DF theorem condition}.

Next we prove, by a constructive manner, that there always exists 
a matrix $T_1$ satisfying 
${\rm Range}(T_1)={\rm Ker}({\cal O}) \cap {\rm Ker}({\cal O}\Sigma_n)$ 
and $T_1^\top\Sigma_n T_1=\Sigma_\ell$. 
First, let us take a vector 
$v_1\in{\rm Ker}({\cal O}) \cap {\rm Ker}({\cal O}\Sigma_n)$, 
which readily implies that $\Sigma_n v_1$ is also contained in 
this space. 
Note that $v_1$ and $\Sigma_n v_1$ are orthogonal. 
Next we take a vector 
$v_2\in{\rm Ker}({\cal O}) \cap {\rm Ker}({\cal O}\Sigma_n)$ 
that is orthogonal to both $v_1$ and $\Sigma_n v_1$. 
Then, as before, $\Sigma_n v_2$ is also an element of this space, 
and further, it is orthogonal to $v_1$, $\Sigma_n v_1$, and $v_2$. 
Repeating the same procedure, we construct 
$T_1=(v_1, \Sigma_n v_1, \ldots, v_\ell, \Sigma_n v_\ell)$; 
this matrix satisfies $T_1^\top T_1=I_{2\ell}$ and 
$\Sigma_n T_1=T_1 \Sigma_\ell$, which thus yield 
$T_1^\top\Sigma_n T_1=\Sigma_\ell$. 
Hence $\hat x_{\rm DF}=T_1^\top\hat x$ satisfies the CCR \eqref{CCR}, 
implying that \eqref{DF theorem condition} leads to the fulfillment of the 
condition (ii) in Def.~\ref{Def of DF mode}.

Consequently, \eqref{DF theorem condition} is an iff condition for the 
system to have a DF subsystem, and $\hat x_{\rm DF}=T_1^\top\hat x$ 
is the DF mode. 
\end{proof}

As proven above, \eqref{DF theorem condition} is not merely 
the equivalent conversion of the condition (i) in Def.~\ref{Def of DF mode}; 
it clarifies the fact that, if there exists an uncontrollable and 
unobservable subspace, it must contain the pair of vectors 
$v$ and $\Sigma_n v$, which directly leads to the concrete 
procedure for constructing the DF subsystem. 
Now we understand that, in Def.~\ref{Def of DF mode}, it is not necessary 
to additionally require the CCR condition (ii); 
that is, if there exists an uncontrollable and unobservable mode, 
it automatically satisfies the CCR.

Below we have an explicit form of the dynamics of the DF and D modes, 
by actually constructing the linear coordinate transformation from 
$\hat x$ to $\hat x'=(\hat x_{\rm DF}^\top,\hat x_{\rm D}^\top)^\top$, 
where $\hat x_{\rm D}$ is the complement mode to $\hat x_{\rm DF}$. 
First, a vector $u_1$ is taken in such a way that it is orthogonal 
to all the column vectors of $T_1$; 
then $\Sigma_n u_1$ is also orthogonal to all the column vectors of 
$T_1$ and $u_1$. 
Second, we take a vector $u_2$ that is orthogonal to both $u_1$ and 
$\Sigma_n u_1$ in addition to the condition $u_2^\top T_1=0$. 
Repeating this procedure, we have 
$T_2=(u_1, \Sigma_n u_1, \ldots, u_{n-\ell}, \Sigma_n u_{n-\ell})$. 
Then, $T=(T_1, T_2)$ is orthogonal, and it satisfies the symplectic 
condition \eqref{symplectic}; i.e. 
\begin{eqnarray*}
& & \hspace*{0em}
     \left( \begin{array}{c} 
                 T_1^\top \\
                 T_2^\top \\
               \end{array}\right)
     (T_1, T_2)
     =I_{2n},
\\ & & \hspace*{0em}
     \left( \begin{array}{c} 
                 T_1^\top \\
                 T_2^\top \\
               \end{array}\right)
     \Sigma_n
     (T_1, T_2)
     =\left( \begin{array}{cc} 
                \Sigma_{\ell} & O \\
                O & \Sigma_{n-\ell} \\
               \end{array}\right). 
\end{eqnarray*}
By construction, $T_2^\top\hat x$ is complement to the uncontrollable 
and unobservable mode $\hat x_{\rm DF}$, hence we can set 
$\hat x_{\rm D}=T_2^\top\hat x$. 
Therefore, the transformation $\hat x'=T^\top \hat x$ yields the 
explicit form of the dynamics of the DF and D modes as follows; 
noting $T_1^\top\Sigma_n=\Sigma_\ell T_1^\top$ and 
$T_2^\top\Sigma_n=\Sigma_{n-\ell} T_2^\top$, 
we find that \eqref{linear QSDE} is transformed to 
\begin{eqnarray}
& & \hspace*{-3em}
\label{DF dynamics}
    d\left( \begin{array}{c} 
                 \hat x_{\rm DF} \\
                 \hat x_{\rm D} \\
               \end{array}\right)
      =\left( \begin{array}{c} 
                 T_1^\top \\
                 T_2^\top \\
               \end{array}\right) A
           (T_1,~T_2)
       \left( \begin{array}{c} 
                 \hat x_{\rm DF} \\
                 \hat x_{\rm D} \\
               \end{array}\right)dt
\nonumber \\ & & \hspace*{5em}
    \mbox{}
       +\left( \begin{array}{c} 
                 T_1^\top \\
                 T_2^\top \\
               \end{array}\right) 
           \Sigma_n C^\top \Sigma_m d\hat{\cal W}
\nonumber \\ & & \hspace*{2.1em}
    =\left( \begin{array}{cc} 
                 \Sigma_\ell T_1^\top GT_1 & O \\
                 O & T_2^\top A T_2 \\
               \end{array}\right)
       \left( \begin{array}{c} 
                 \hat x_{\rm DF} \\
                 \hat x_{\rm D} \\
               \end{array}\right)dt
\nonumber \\ & & \hspace*{5em}
    \mbox{}
       +\left( \begin{array}{c} 
                 O \\
                 \Sigma_{n-\ell} T_2^\top C^\top \Sigma_m  \\
            \end{array}\right) d\hat{\cal W}, 
\end{eqnarray}
where the relation $CT_1=O$ is used. 
Also, $T_2^\top AT_1=O$ follows from the condition ${\cal O}T_1=O$ 
($\Leftrightarrow CA^kT_1=O,~\forall k\geq 0$), which implies that 
there exists a matrix $P$ satisfying $AT_1=T_1 P$. 
Moreover, the output equation \eqref{linear output} is represented 
as 
\begin{eqnarray}
& & \hspace*{-3em}
\label{DF output}
   d\hat{\cal W}^{\rm out}
    =C(T_1,~T_2)
        \left( \begin{array}{c} 
                 \hat x_{\rm DF} \\
                 \hat x_{\rm D} \\
               \end{array}\right)dt + d\hat{\cal W}
\nonumber \\ & & \hspace*{0.1em}
    =C T_2 \hat x_{\rm D}dt + d\hat{\cal W}. 
\end{eqnarray}
From \eqref{DF dynamics} and \eqref{DF output}, $\hat x_{\rm DF}$ 
is apparently a DF mode that is completely isolated from the 
environmental noise fields. 
In particular, the equation 
$d\hat x_{\rm DF}/dt = \Sigma_\ell T_1^\top G T_1 \hat x_{\rm DF}$ 
means that the DF mode obeys the unitary time evolution governed by 
the Hamiltonian 
\begin{equation}
\label{DF Hamiltonian}
    \hat H_{\rm DF}
      =\hat x_{\rm DF}(0)^\top G_{\rm DF} \hat x_{\rm DF}(0)/2,~~~
    G_{\rm DF}=T_1^\top G T_1. 
\end{equation}
That is, the unitary gate operation 
$\hat U_{\rm DF}(t)={\rm exp}(-it\hat H_{\rm DF})$ can be carried 
out coherently on the DF mode $\hat x_{\rm DF}$.


\subsection{DF conditions for Hamiltonian engineering}

Lastly in this section, we consider the following practical 
question; 
given a system-environment interaction such that there exists a DF subsystem, 
what kind of system Hamiltonian is allowed in order for the DF subsystem 
to be unchanged? 
More precisely, if the condition \eqref{DF theorem condition} is 
satisfied for a given $C$ and $G=O$, how can we construct a 
non-zero matrix $G$ such that \eqref{DF theorem condition} 
still holds? 
Here a simple characterization of such $G$ is given:

\begin{proposition}
\label{DF mode prop}
Suppose that, for a given $C$, there exists a non-zero matrix 
$T_1$ satisfying 
${\rm Range}(T_1)={\rm Ker}(C) \cap {\rm Ker}(C\Sigma_n)$ 
and $\Sigma_n T_1=T_1 \Sigma_\ell$. 
Then, if ${\rm Range}(T_1)$ is invariant under a linear transformation 
$G$, the system specified by this $G$ and the same $C$ has a DF subsystem 
with mode $T_1^\top \hat x$. 
\end{proposition}

\begin{proof} 
The invariance condition means that there exists a matrix $Q$ satisfying 
$GT_1=T_1 \Sigma_\ell^\top Q$, which is equivalent to 
$GT_1=\Sigma_n^\top T_1 Q$ due to the assumption on $T_1$; 
thus $(\Sigma_n G)T_1=T_1 Q$ holds. 
Then, since $CT_1=O$, we have $C(\Sigma_n G)^k T_1=O$ 
for all $k\geq 0$. 
Again from the relation $\Sigma_n T_1=T_1 \Sigma_\ell$, this 
implies $C(\Sigma_n G)^k \Sigma_n T_1=O,~\forall k\geq 0$ 
as well. 
Consequently, it follows from a similar calculation shown in the 
proof of Lemma~\ref{DF lemma} that 
${\cal O}T_1={\cal O}\Sigma_n T_1=0$ is satisfied; 
this means the system has a DF subsystem with DF mode 
$T_1^\top \hat x$. 
\end{proof}

This proposition further leads to a convenient result in 
a particular case:

\begin{corollary}
\label{DF mode corollary}
In addition to the assumptions made in Proposition \ref{DF mode prop}, 
suppose ${\rm Ker}(C) = {\rm Ker}(C\Sigma_n)$. 
Then, the system specified by $G$ and the same $C$ has a DF subsystem 
with mode $T_1^\top \hat x$, if and only if $CGT_1=O$. 
\end{corollary}

\begin{proof} 
As proven above, $T_1^\top \hat x$ is a DF mode if and only if 
$C(\Sigma_n G)^k T_1=O$ holds for all $k\geq 0$. 
Then, because of ${\rm Range}(T_1)={\rm Ker}(C)$, it is equivalent 
to $(\Sigma_n G)T_1=T_1 Q,~\exists Q$. 
Hence, we further equivalently have $GT_1=T_1\Sigma_\ell^\top Q$ 
and thus $CGT_1=O$. 
\end{proof}


\section{Stability of DF subsystem}

In this section, we assume that the system's initial state is 
Gaussian, hence the state is always Gaussian and is characterized 
by only the mean vector and the covariance matrix, as explained 
in Section~\ref{Preliminary: Gaussian states}. 
Within this framework, we here study a stability property of the general linear 
DF subsystem. 
Note that in the finite dimensional case some useful results 
guaranteeing similar stability properties of a DF subsystem have been 
obtained in \cite{SchirmerPRA2010,Ticozzi2008,Ticozzi2009}.

In order to coherently manipulate a quantum state of the DF subsystem, 
i.e. $\hat \rho_{\rm DF}$, it has to be separated from that of the D 
subsystem with mode $\hat x_{\rm D}$; 
actually, if $\hat x_{\rm DF}$ is quantally correlated (i.e., 
entangled) with $\hat x_{\rm D}$, this means that $\hat \rho_{\rm DF}$ 
can be affected from the environment. 
This requirement for separability of the DF and the D states is, 
in Gaussian case, implied by the condition that the covariance 
matrix $V$ of the whole system is expressed as 
$V={\rm diag}\{V_{\rm DF}, V_{\rm D}\}$, where $V_{\rm DF}$ and 
$V_{\rm D}$ correspond to the covariance matrices of $\hat x_{\rm DF}$ 
and $\hat x_{\rm D}$, respectively. 
Now, from \eqref{mean and variance dynamics}, the covariance 
matrix of the system \eqref{DF dynamics} changes in time with 
the following Lyapunov differential equation:
\begin{equation}
\label{DF Lyapunov}
   \frac{dV}{dt}
    = \left( \begin{array}{cc} 
               A_1 & O \\
               O & A_2 \\
             \end{array}\right)V
     +V\left( \begin{array}{cc} 
               A_1^\top & O \\
               O & A_2^\top \\
             \end{array}\right)
     +\left( \begin{array}{cc} 
               O & O \\
               O & D_2 \\
             \end{array}\right), 
\end{equation}
where $A_1=\Sigma_\ell G_{\rm DF}$, $A_2=T_2^\top AT_2$, and 
$D_2=T_2^\top\Sigma_n^\top C^\top C\Sigma_n T_2/2$. 
Hence, if the solution of \eqref{DF Lyapunov} takes a block diagonal form 
$V(t)={\rm diag}\{V_{\rm DF}(t), V_{\rm D}(t)\}$ in a long time limit 
without respect to the initial value $V(0)$, this means that the DF 
dynamics is robust in the sense that an unwanted correlation between 
the DF and the D modes autonomously decreases and finally vanishes. 
The following theorem provides a convenient criterion for the 
system to have this desirable property (${\rm eig}(A)$ denotes 
any eigenvalue of a matrix $A$).

\begin{theorem}
\label{DF stability}
$\Re[{\rm eig}(\Sigma_\ell G_{\rm DF}) + {\rm eig}(A_2)]<0$ if and 
only if in a long time limit the solution of \eqref{DF Lyapunov} 
takes a form $V(t)={\rm diag}\{V_{\rm DF}(t), V_{\rm D}(t)\}$, where 
$V_{\rm DF}(t)$ and $V_{\rm D}(t)$ correspond to the covariance 
matrices of $\hat x_{\rm DF}(t)$ and $\hat x_{\rm D}(t)$, respectively. 
Moreover, if $G_{\rm DF}\geq 0$, then the above iff condition is 
simplified to $\Re[{\rm eig}(A_2)]<0$. 
\end{theorem}

\begin{proof} 
When partitioning the matrix variable $V$ as 
$V=(V_1, V_2 ; V_2^\top, V_3)$, the dynamics of $V_2$ is given by 
\[
    dV_2/dt = A_1V_2 + V_2^\top A_2^\top. 
\]
We now take a vector form of this matrix differential equation. 
For this purpose, let $v_2$ be a collection of the row vectors 
of $V_2$, which is a real $4\ell(n-\ell)$-dimensional row vector. 
Then $v_2$ obeys the dynamics of the form $dv_2/dt=\bar{A}v_2$ 
where $\bar{A}=A_1\otimes I_{2(n-\ell)} + I_{2\ell}\otimes A_2$. 
Now note that any eigenvalue of $\bar{A}$ is equal to 
${\rm eig}(A_1) + {\rm eig}(A_2)$. 
Therefore, we have $v_2\rightarrow 0$ as $t\rightarrow \infty$, 
if and only if the real parts of sum of any eigenvalues of 
$A_1=\Sigma_\ell G_{\rm DF}$ and $A_2$ are strictly negative.

To show the second part of the theorem, let us consider the 
eigen-equation $\Sigma_\ell G_{\rm DF} g=\lambda g$ with $g$ the 
eigenvector and $\lambda$ the corresponding eigenvalue. 
From this we have $g^\dagger (G_{\rm DF}\Sigma_\ell G_{\rm DF})g
=\lambda g^\dagger G_{\rm DF}g$, which immediately leads to 
$(\lambda+\lambda^*)g^\dagger G_{\rm DF}g=0$. 
If $G_{\rm DF}\geq 0$, this can be further expressed as 
$(\lambda+\lambda^*)\|\sqrt{G_{\rm DF}}g\|^2=0$, implying 
$\lambda+\lambda^*=0$ or $G_{\rm DF}g=0$; 
in both cases we have $\Re[{\rm eig}(\Sigma_\ell G_{\rm DF})]=0$. 
\end{proof}

We remark that the specific assumption $G_{\rm DF}> 0$ is significant 
from the viewpoint of state manipulation; 
recall that the unitary gate operation with a quadratic Hamiltonian 
$\hat H=\hat x^\top G\hat x/2$ is given by 
$\hat U(t)={\rm exp}(-it \hat x^\top G\hat x/2)$. 
Then, it was shown in \cite{BurgarthPRL2012} that, if $G> 0$, for 
any $\epsilon>0$ and $\tau_1<0$ there always exists $\tau_2>0$ 
such that $\|\hat U(\tau_1)-\hat U(\tau_2)\|<\epsilon$ with 
$\| \bullet \|$ an appropriate operator norm. 
This result means that a unitary operation requiring propagation 
with negative time can be always realized by a physically realizable 
unitary operation with positive time; 
this condition is indeed necessary to implement a desirable unitary 
gate operation via a sequence of unitaries generated from a set of 
available Hamiltonians.


\section{Coherent state manipulation and preservation on a DF subsystem}

Coherent state manipulation and preservation are certainly 
primary applications of a DF subsystem, as demonstrated extensively 
in the finite dimensional case. 
This section presents two simple examples to show how a DF mode 
can be used in order to achieve these goals in an infinite 
dimensional linear quantum system.


\subsection{Particles trapped via dissipative coupling}
\label{DFmodeExample1}

\begin{figure}
\centering
\includegraphics[scale=0.45]{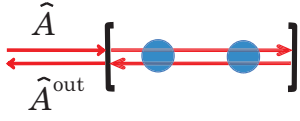}
\caption{
\label{DFmodeExample}
Two trapped particles in a two-sided optical cavity. 
}
\end{figure}

Let us consider a toy system composed of two ``particles" trapped in a two-sided 
optical cavity depicted in Fig.~\ref{DFmodeExample}. 
The coupling operator between the particles and the outer 
environment field is $\hat L=\sqrt{\kappa}(\hat b_1 + \hat b_2)$, 
where $\hat b_i=(\hat q_i + i\hat p_i)/\sqrt{2}$; 
$\hat q_i$ and $\hat p_i$ represent the position and momentum 
operators of the $i$-th particle. 
This dissipation effect described by $\hat L$ stems from the 
coupling of the particles to the single-mode cavity field that 
immediately decays and can be adiabatically eliminated 
\cite{WinelandRMP2003}; 
the parameter $\kappa$ represents the decay rate of this dissipation. 
For more realistic setups, see \cite{Dong2012,Huang2013,Clerk2012}. 
The system dynamics \eqref{linear QSDE} and the output equation 
\eqref{linear output} are now given by 
\[
    d\hat x=A\hat x dt + Bd\hat{\cal W},~~~
    d\hat{\cal W}^{\rm out}=C\hat xdt + d\hat{\cal W},
\]
where
\begin{eqnarray}
& & \hspace*{0em}
    A=\Sigma_2 G 
       - \frac{\kappa}{2}
        \left( \begin{array}{cc} 
                 I_2 & I_2 \\
                 I_2 & I_2 \\
               \end{array}\right),~~~
    B=-\sqrt{\kappa}
        \left( \begin{array}{c} 
                 I_2 \\
                 I_2 \\
               \end{array}\right),~~~
\nonumber \\ & & \hspace*{0em}
    C = \sqrt{\kappa}
           \left( \begin{array}{cc} 
                 I_2 & O  \\
                 O & I_2  \\
               \end{array}\right). 
\nonumber 
\end{eqnarray}

We begin with the assumption that the system does not have its 
own Hamiltonian (i.e., $G=0$) and then try to find a DF subsystem. 
In this case, we have 
\begin{eqnarray}
& & \hspace*{0em}
    {\cal O}=C
      = \sqrt{\kappa}
           \left( \begin{array}{cccc} 
                 1 & 0 & 1 & 0 \\
                 0 & 1 & 0 & 1 \\
               \end{array}\right),
\nonumber \\ & & \hspace*{0em}
    {\cal O}\Sigma_2=C\Sigma_2
      = \sqrt{\kappa}
           \left( \begin{array}{cccc} 
                 0 & 1 & 0 & 1 \\
                 -1 & 0 & -1 & 0 \\
               \end{array}\right). 
\nonumber
\end{eqnarray}
Therefore, $v_1=(1, 0, -1, 0)^\top/\sqrt{2}$ and 
$\Sigma_2 v_1=(0, 1, 0, -1)^\top/\sqrt{2}$ span the kernels 
of both of these matrices; 
i.e., ${\rm span}\{v_1, \Sigma_2 v_1 \}={\rm Ker}({\cal O})
={\rm Ker}({\cal O}\Sigma_2)$. 
Hence, from Theorem \ref{DF mode theorem} the system has 
a DF subsystem. 
In particular, the procedure shown in the proof of 
Theorem~\ref{DF mode theorem} yields $T_1=(v_1, \Sigma_2 v_1)$, 
thus the DF mode is given by 
\[
   \hat x_{\rm DF} = T_1^\top \hat x
       = \frac{1}{\sqrt{2}}
         \left( \begin{array}{c} 
                 \hat q_1 - \hat q_2 \\
                 \hat p_1 - \hat p_2 \\
               \end{array}\right). 
\]

Here let us discuss adding a nonzero matrix $G$ that does not 
change the above-obtained DF subsystem. 
Now ${\rm Ker}({\cal O})={\rm Ker}({\cal O}\Sigma_2)$, hence 
from Corollary~\ref{DF mode corollary}, $G$ is required 
to satisfy $CGT_1=O$. 
Then it is immediate to see that $G$ has to be of the form 
$G=(G_1, G_2 ; G_2, G_1)$ where $G_1$ and $G_2$ are both real 
$2 \times 2$ symmetric matrices. 
Therefore, the DF subsystem is driven by the 
Hamiltonian~\eqref{DF Hamiltonian} 
with $G_{\rm DF}=T_1^\top G T_1=G_1-G_2$. 
This means that, if each particle is governed by the Hamiltonian specified 
by $G_1$ and this Hamiltonian can be arbitrarily chosen, we can perform 
arbitrary linear gate operation on the state of $\hat x_{\rm DF}$. 
Also, when $G_{\rm DF}=G_1-G_2=O$, we can preserve 
any state of the DF subsystem.

Next, to evaluate the stability of the DF subsystem, let us explicitly 
construct the dynamics of both $\hat x_{\rm DF}$ and $\hat x_{\rm D}$, 
according to the procedure shown above \eqref{DF dynamics}. 
To obtain the matrix $T_2$, we chose $u_1=(1, 0, 1, 0)^\top/\sqrt{2}$ 
as a basis vector orthogonal to $v_1$ and $\Sigma_2 v_1$. 
Then $\Sigma_2 u_1=(0, 1, 0, 1)^\top/\sqrt{2}$ is orthogonal to 
$v_1, \Sigma_2 v_1$, and $u_1$, implying $T_2=(u_1, \Sigma_2 u_1)$. 
Hence, the linear coordinate transformation is obtained as 
\begin{equation}
\label{DFmodeExample T}
    T=(T_1, T_2),~~
    T_1=\frac{1}{\sqrt{2}}
        \left( \begin{array}{c} 
                 I_2  \\
                 -I_2 \\
               \end{array}\right),~~
    T_2=\frac{1}{\sqrt{2}}
        \left( \begin{array}{c} 
                 I_2 \\
                 I_2 \\
               \end{array}\right). 
\end{equation}
The transformed variable $\hat x'=T^\top \hat x$ then satisfies
\[
    d\hat x'=A'\hat x' dt + B'd\hat{\cal W},~~~
    d\hat{\cal W}^{\rm out}=C'\hat x'dt + d\hat{\cal W},
\]
where
\begin{eqnarray*}
& & \hspace*{-0.5em}
    A'= \left( \begin{array}{cc} 
                 \Sigma(G_1-G_2) & O \\
                 O & \Sigma(G_1+G_2)-\kappa I_2 \\
               \end{array}\right),
\nonumber \\ & & \hspace*{-0.5em}
    B'=-\sqrt{2\kappa}
        \left( \begin{array}{c} 
                 O \\
                 I_2 \\
               \end{array}\right),~~
    C'=\sqrt{2\kappa}(O,~I_2). 
\end{eqnarray*}
As expected, the first CCR pair of $\hat x'$, i.e., $\hat x_{\rm DF}$, 
is neither affected by the incoming noise field nor 
observed from the output field. 
Now assume that the state is Gaussian. 
Then, by taking large $\kappa$ so that the condition of 
Theorem~\ref{DF stability} is satisfied, we can guarantee the 
stability of the DF subsystem. 
That is, the off-diagonal block matrix of the covariance matrix 
$V'=\mean{\Delta \hat x' \Delta \hat x'\mbox{}^\top 
 + (\Delta \hat x' \Delta \hat x'\mbox{}^\top)^\top}/2$ 
converges to zero as $t\rightarrow \infty$, thus in the long time 
limit the states of DF and D modes are always separated. 
Hence, with this Gaussian system, arbitrary linear gate operation 
and state preservation can be carried out in a robust manner, 
under relatively mild assumptions.

Lastly, under the assumption that the state is Gaussian and the 
condition $G_{\rm DF}=G_1-G_2=O$, let us examine the state preserved 
in the whole system. 
In fact, in this case $\hat x_{\rm DF}$ does not change in time, and hence 
the whole state of $\hat x$ is preserved if $\hat x_{\rm D}$ is in a steady 
state. 
This condition is achieved by again taking large $\kappa$ so 
that the coefficient matrix of the drift term of $\hat x_{\rm D}$ 
is Hurwitz; 
in this case, $\hat x_{\rm D}$ has a steady Gaussian state with 
covariance matrix $I_2/2$. 
As a result, the covariance matrix of the whole state of $\hat x'$ 
is $V'(0)={\rm diag}\{ V_{\rm DF}(0), I_2/2 \}$, where here the 
initial time is reset to $t=0$. 
The covariance matrix of $\hat x$ is then given by 
\begin{eqnarray*}
& & \hspace*{-0.5em}
    V(0)=T V'(0) T^\top 
\nonumber \\ & & \hspace*{1.8em}
       = \frac{1}{2}
               \left( \begin{array}{cc} 
                  I_2/2 + V_{\rm DF}(0) & I_2/2 - V_{\rm DF}(0) \\
                  I_2/2 - V_{\rm DF}(0) & I_2/2 + V_{\rm DF}(0) \\
               \end{array}\right). 
\end{eqnarray*}
If the initial state of $\hat x_{\rm DF}$ can be set to a 
{\it squeezed state} 
\cite{BraunsteinRMP,Furusawa2011,GardinerBook,WallsBook,WisemanBook} 
with covariance matrix 
$V_{\rm DF}(0)={\rm diag}\{ e^r/2, e^{-r}/2 \}$, then $V(0)$ 
takes the form
\begin{equation}
\label{state preservation example}
    V(0)
       = \frac{1}{4}
               \left( \begin{array}{cc|cc} 
                  1 + e^r &  & 1 - e^r &  \\
                   & 1 + e^{-r} &  & 1 - e^{-r} \\ \hline
                  1 - e^r &  & 1 + e^r &  \\
                   & 1 - e^{-r} &  & 1 + e^{-r} \\
               \end{array}\right).
\end{equation}
As mentioned above, this state does not change in time. 
Thus the state preserved is a pure entangled Gaussian state called the 
{\it two-mode squeezed state} \cite{BraunsteinRMP,Furusawa2011}; 
actually, the purity for this state is $P=1/\sqrt{16{\rm det}(V(0))}=1$, 
and the logarithmic negativity \cite{VidalPRA2002}, which is a 
convenient measure for entanglement, is $E_{\cal N}=r/2>0$. 
This result was found in a specific setup \cite{PrauznerJPA2004}, but 
here we have the following general statement; 
in general, preservation of a nontrivial state of a DF subsystem can 
lead to that of an entangled state of the whole system, if the 
stability of the DF mode is guaranteed. 
In other words, a system having a stable DF mode can be utilized 
for preserving an entangled state among the whole system, as well 
as a state of the DF subsystem.


\subsection{Particles trapped via dispersive coupling}
\label{DFmodeExample2}

The second example shown here deals with the two particles trapped 
in an optical cavity as before, but with different coupling 
$\hat L=\sqrt{\kappa/2}(\hat q_1 + \hat q_2)$, where again $\hat q_i$ 
is the position operator of the $i$-th particle. 
This coupling realizes a continuous-time measurement of 
$\hat q_1 + \hat q_2$; 
see e.g. \cite{DohertyPRA1998,DohertyPRA1999} for how to actually 
construct this kind of position monitoring scheme. 
Now $c=\sqrt{\kappa/2}(1, 0 , 1, 0)$, hence, when $G=0$, we have 
\[
    {\cal O}
      = \sqrt{\kappa}
           \left( \begin{array}{cccc} 
                 1 & 0 & 1 & 0 \\
                 0 & 0 & 0 & 0 \\
               \end{array}\right), ~~
    {\cal O}\Sigma_2
      = \sqrt{\kappa}
           \left( \begin{array}{cccc} 
                 0 & 1 & 0 & 1 \\
                 0 & 0 & 0 & 0 \\
               \end{array}\right), 
\]
and thus
\begin{eqnarray*}
& & \hspace*{0em}
\small{
    {\rm Ker}({\cal O})
     ={\rm span}\Big\{
        \frac{1}{\sqrt{2}}
        \left( \begin{array}{c} 
              1 \\
              0 \\
              -1 \\
              0 \\
          \end{array}\right),
        \left( \begin{array}{c} 
              0 \\
              1 \\
              0 \\
              0 \\
          \end{array}\right),
        \left( \begin{array}{c} 
              0 \\
              0 \\
              0 \\
              1 \\
          \end{array}\right) \Big\},
}
\\ & & \hspace*{0em}
\small{
     {\rm Ker}({\cal O}\Sigma_2)
     ={\rm span}\Big\{
        \left( \begin{array}{c} 
              1 \\
              0 \\
              0 \\
              0 \\
          \end{array}\right),
        \left( \begin{array}{c} 
              0 \\
              0 \\
              1 \\
              0 \\
          \end{array}\right), 
        \frac{1}{\sqrt{2}}
        \left( \begin{array}{c} 
              0 \\
              1 \\
              0 \\
              -1 \\
          \end{array}\right) \Big\}. 
}
\end{eqnarray*}
The intersection of these two spaces is the same as that obtained 
in the previous example, thus we have the same coordinate 
transformation matrix \eqref{DFmodeExample T}. 
However, the dynamical equation of $\hat x'=T^\top \hat x$ differs from 
the previous one: 
\[
    d\hat x'=A'\hat x' dt + B'd\hat{\cal W},~~~
    d\hat{\cal W}^{\rm out}=C'\hat x'dt + d\hat{\cal W},
\]
where
\begin{eqnarray*}
& & \hspace*{-2em}
    A'= \left( \begin{array}{cc} 
                 \Sigma(G_1-G_2) & O \\
                 O & \Sigma(G_1+G_2) \\
               \end{array}\right),
\nonumber \\ & & \hspace*{-2em}
    B'=-\sqrt{2\kappa}
        \left( \begin{array}{c} 
                 O \\
                 E_{22} \\
               \end{array}\right),~~
    C'=\sqrt{2\kappa}(O,~E_{11}). 
\end{eqnarray*}
Here $E_{11}={\rm diag}\{1, 0\}$ and $E_{22}={\rm diag}\{0, 1\}$. 
Note that $\hat q'_2=(\hat q_1+\hat q_2)/\sqrt{2}$ is not 
affected by the noisy environment while it appears in the output 
equation. 
This means that we can extract information about $\hat q'_2$ without 
disturbing it; that is, $\hat q'_2$ is a {\it quantum non-demolition} 
(QND) observable (see e.g. \cite{WisemanBook}).

Let us now discuss the state manipulation and preservation with 
this system. 
As in the previous example, if two particles have the same and 
arbitrary Hamiltonian specified by $G_1$, this means that 
any linear gate operation on the DF mode is possible. 
However, this system does not satisfy the condition of 
Theorem~\ref{DF stability}, implying that the system does not 
have capability of decoupling the DF and the D modes once they 
become correlated. 
In this sense, the state manipulation with this DF mode is 
not robust. 
Next, to see if the system can preserve a state of the whole 
system, let us again set the initial state to the Gaussian state 
with covariance matrix \eqref{state preservation example}. 
A notable difference from the previous case is that now the D mode 
does not have a steady state, and eventually the whole state is 
a time-varying one with covariance matrix 
\[
    V(0) \rightarrow 
    V(t) = V(0) + \frac{1}{4}
               \left( \begin{array}{cc} 
                  \kappa t E_{22} & \kappa t E_{22} \\
                  \kappa t E_{22} & \kappa t E_{22} \\
               \end{array}\right). 
\]
The purity of this Gaussian state is 
$P=1/\sqrt{16{\rm det}(V(t))}=1/\sqrt{1+\kappa t}$ and thus 
converges to zero as $t\rightarrow \infty$. 
That is, any useful state does not remain alive in the long time 
limit. 
As a result, due to the lack of the stability property, this 
system offers neither desirable state manipulation nor 
preservation.


\section{Engineering a DF linear subsystem}
\label{Synthesis of a DF linear system}

This section deals with the topic of how to {\it engineer} a DF 
subsystem, when a certain linear open system is given to us. 
More precisely, the problem is not how to find a DF component in 
that system, but rather how to synthesize an auxiliary system so 
that the extended system has a desirable DF mode.

This engineering problem is important for purely exploring the 
nature of a system of interest, as well as for the application 
to quantum information science such as coherent state manipulation 
and preservation as discussed in the previous section. 
In fact, any quantum system is essentially an open system and 
its properties must be perturbed by surrounding environment. 
Therefore, constructing a subsystem that is exactly isolated from 
environment is of vital importance.


\subsection{Opto-mechanical oscillator}

\begin{figure}
\centering
\includegraphics[scale=0.45]{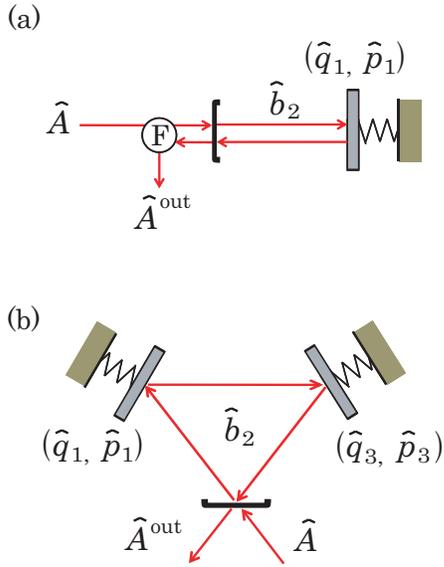}
\caption{
\label{Oscillators}
(a) An opto-mechanical oscillator coupled to an optical cavity field. 
$F$ denotes a Faraday isolator that facilitates one-way coupling 
between the input and cavity fields. 
(b) Two opto-mechanical oscillators coupled to a common optical 
cavity field. 
In both cases, the optical cavity field pushes the oscillator and changes 
its motion through the radiation pressure force. 
}
\end{figure}

The first example is an opto-mechanical oscillator depicted in 
Fig.~\ref{Oscillators}~(a); 
see e.g. \cite{LawPRA1995} for a more detailed description. 
The system is composed of a two-sided optical cavity where one of 
the mirrors is movable and acts as an oscillator. 
Let $\hat q_1$ and $\hat p_1$ be the oscillator's position and 
momentum operators, respectively. 
Also define $\hat q_2=(\hat b_2 + \hat b_2^*)/\sqrt{2}$ and 
$\hat p_2=(\hat b_2 - \hat b_2^*)/\sqrt{2}i$ with $\hat b_2$ the 
annihilation operator of the cavity mode. 
The oscillator is driven by the free Hamiltonian 
$\hat H_0=m\omega^2\hat q_1^2+\hat p_1^2/m$ where $m$ and $\omega$ 
represent the mass and the resonant frequency of the oscillator, 
and further, it is coupled to the cavity mode through the radiation 
pressure Hamiltonian 
$\hat H_{\rm rp}=\epsilon \hspace{0.05cm}\hat q_1 \hat b_2^*\hat b_2$; 
consequently the whole system Hamiltonian is given by 
$\hat H=\hat H_0 + \hat H_{\rm rp}$. 
The cavity mode couples to an outer coherent light field at the 
other mirror in a standard dissipative manner, corresponding to 
$\hat L=\sqrt{\kappa}\hat b_2$. 
As a result, the linearized dynamics of the whole oscillator-cavity 
system is given by
\[
    d\hat x=A\hat x dt + Bd\hat{\cal W},~~~
    d\hat{\cal W}^{\rm out}=C\hat xdt + d\hat{\cal W},
\]
where
\begin{eqnarray*}
& & \hspace*{-3em}
    A=\left( \begin{array}{cc|cc} 
                 0 & 1/m & 0 & 0 \\
                 -m\omega^2 & 0 & \gamma & 0 \\ \hline
                 0 & 0 & -\kappa & 0 \\
                 \gamma & 0 & 0 & -\kappa 
               \end{array}\right),
\nonumber \\ & & \hspace*{-3em}
    B=-\sqrt{2\kappa}
        \left( \begin{array}{cc} 
                 0 & 0 \\
                 0 & 0 \\ \hline
                 1 & 0 \\
                 0 & 1 \\
               \end{array}\right),~~~
    C=-B^\top.
\end{eqnarray*}
Here $\gamma$ denotes the radiation pressure strength proportional 
to $\epsilon$ and $\kappa$ denotes the decay rate of the cavity field.

It is immediate to see that ${\cal O}$ is of full rank, and 
thus the system does not have a DF subsystem. 
Hence, as posed in the beginning of this section, we try to 
engineer an appropriate extended system so that it contains a 
DF subsystem. 
In particular, we follow the idea of Tsang \cite{Tsang2010}; 
an auxiliary system with single mode $(\hat q_3, \hat p_3)$ is 
prepared and directly coupled to the cavity mode through 
the interaction Hamiltonian $\hat H_{\rm int}=-g\hat q_2 \hat q_3$ 
with $g$ the coupling strength. 
Then the extended system of variable 
$\hat x_e=(\hat x^\top, \hat q_3, \hat p_3)^\top$ is given by 
\[
    d\hat x_e=A_e\hat x_e dt + B_e d\hat{\cal W},~~~
    d\hat{\cal W}^{\rm out}=C_e\hat x_e dt + d\hat{\cal W},
\]
where
\begin{eqnarray*}
& & \hspace*{-1em}
       A_e = \left( \begin{array}{c|c} 
                 A &  \begin{array}{cc} 
                         0 & 0 \\ 
                         0 & 0 \\
                         0 & 0 \\
                         g & 0 \\
                       \end{array}   \\ \hline
                 \begin{array}{cccc} 
                    0 & 0 & 0 & 0 \\ 
                    0 & 0 & g & 0 \\
                 \end{array}          & 
                       \begin{array}{cc} 
                         0 & \mu \\ 
                         \nu & 0 \\
                       \end{array}
               \end{array}\right),
\nonumber \\ & & \hspace*{-1em}
        B_e = \left( \begin{array}{c} 
                 B \\ \hline
                 O_2 
              \end{array}\right),~~~
        C_e = -B_e^\top. 
\end{eqnarray*}
Here the auxiliary system is assumed to have the Hamiltonian 
$(-\nu\hat q_3^2+\mu\hat p_3^2)/2$. 
Also $O_2$ denotes the $2\times 2$ zero matrix. 
The problem is to find a set of parameters $(g,\mu,\nu)$ such 
that the above extended system has a DF component 
\footnote{
The problem considered in \cite{Tsang2010} is to find the 
parameter set $(g, \mu, \nu)$ so that the input noise field 
$\hat Q(t)$ does not appear in the output field 
$\hat P^{\rm out}(t)$, for the purpose of enhancing the sensitivity 
to the incoming unknown force acting on the oscillator. 
}.
First, we need that the observability matrix ${\cal O}$ composed 
of $A_e$ and $C_e$ is not of full rank; 
this requirement yields $\mu\nu=-\omega^2$. 
Next, the iff condition given in Theorem~\ref{DF mode theorem} 
imposes us to have $\mu=1/m$. 
As a result, the parameters must satisfy $\mu=1/m$ and $\nu=-m\omega^2$, 
implying that the Hamiltonian of the auxiliary system is exactly 
the same as that of the original oscillator. 
Although we can fabricate this auxiliary system using only some 
optical devices, a natural situation is that another opto-mechanical 
oscillator with the same mass and resonant frequency serves as the 
auxiliary system. 
The resulting extended system is shown in Fig.~\ref{Oscillators}~(b), 
which has been investigated recently in \cite{Huang2013} in the 
context of quantum memory. 
Note that the coupling strength $g$ can be chosen arbitrarily; 
this strikingly differs from the case \cite{Tsang2010} where $g$ 
is also required to satisfy a certain condition.

The dynamical equations of the DF and the D modes are explicitly 
obtained as follows. 
With the parameters determined above, the procedure shown in 
Section~III yields the linear coordinate transformation matrix 
$T=(T_1, T_2)$ as 
\[
     T_1=\frac{1}{\gamma'}
               \left( \begin{array}{c} 
                 g I_2 \\
                 O_2 \\ 
                 -\gamma I_2 \\
               \end{array}\right),~~~
     T_2=\frac{1}{\gamma'}
               \left( \begin{array}{cc} 
                 \gamma I_2 & O_2 \\
                 O_2 & \gamma' I_2 \\ 
                 g I_2 & O_2 \\
               \end{array}\right), 
\]
where $\gamma'=\sqrt{\gamma^2+g^2}$. 
Hence, the DF mode is 
\[
    \hat x_{\rm DF} = T_1^\top \hat x_e
       = \frac{1}{\gamma'}
         \left( \begin{array}{c} 
                 g \hat q_1 - \gamma \hat q_3 \\
                 g \hat p_1 - \gamma \hat p_3 \\
               \end{array}\right), 
\]
and the dynamics of $\hat x_e'=T^\top \hat x_e = 
(\hat x_{\rm DF}^\top, \hat x_{\rm D}^\top)^\top$ is given by
\[
    d\hat x'_e=A'_e\hat x'_e dt + B'_e d\hat{\cal W},~~~
    d\hat{\cal W}^{\rm out}=C'_e\hat x_e dt + d\hat{\cal W},
\]
where
\begin{eqnarray*}
& & \hspace*{-1em}
       A'_e = \left( \begin{array}{c|c} 
                 \begin{array}{cc} 
                         0 & 1/m \\ 
                         -m\omega^2 & 0 \\
                       \end{array}    &    \\ \hline
                  & A 
               \end{array}\right),
\nonumber \\ & & \hspace*{-1em}
        B'_e = \left( \begin{array}{c} 
                 O_2 \\ \hline
                 B
              \end{array}\right),~~~
        C'_e = -B'_e\mbox{}^\top. 
\end{eqnarray*}
Here now the parameter $\gamma$ in $A$ is replaced by 
$\gamma'=\sqrt{\gamma^2+g^2}$. 
This equation shows that the relative coordinate of the oscillators 
is decoupled from the cavity field; 
actually, if the radiation pressure force pushes the oscillators along 
the same direction, then their relative coordinate is independent to 
the optical path length of the cavity, so it should be decoherence free. 
Note that, of course, this can happen only when the two oscillators are 
identical. 
A notable point is that the DF subsystem is governed by the same 
Hamiltonian as that of the original oscillator. 
This means that an ideal closed opto-mechanical oscillator 
can be in fact implemented in a realistic open system.


\subsection{Trapped particles}

\begin{figure}
\centering
\includegraphics[scale=0.5]{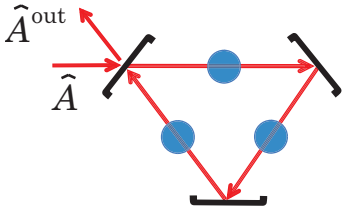}
\caption{
\label{three particles cavity}
Three particles trapped in a ring-type cavity. 
}
\end{figure}

Let us again consider the system of two particles discussed in 
Section~\ref{DFmodeExample1}, which is now additionally subjected 
to the following Hamiltonian: 
\begin{equation}
\label{example original}
   \hat H=\frac{1}{2}\sum_{j=1,2}\big(
              \frac{1}{m}\hat p_j^2 + m\omega^2 \hat q_j^2 \big)
                 +\frac{1}{2}k(\hat q_1-\hat q_2)^2,
\end{equation}
where $m$ and $\omega$ are the mass and the resonant frequency 
of both of the particles. 
The particles interact with each other through a harmonic potential 
with strength $k$, as indicated in the second term of \eqref{example original}. 
The corresponding $G$ matrix is of the form $G=(G_1, G_2 ; G_2, G_1)$, 
thus from the result obtained in Section~\ref{DFmodeExample1} this 
system has a DF component. 
However, this single-mode DF subsystem cannot, of course, simulate 
the behavior of the closed two-mode system driven by the 
Hamiltonian \eqref{example original}. 
Therefore, let us try to engineer an extended system having a 
two-mode DF subsystem in which the closed dynamics driven by the 
Hamiltonian \eqref{example original} is realized. 
That is, we want to simulate the closed system with Hamiltonian 
\eqref{example original}, in an open quantum system.

Let us consider an extended system composed of three particles 
trapped in a ring-type cavity depicted in Fig.~\ref{three particles 
cavity}. 
The Hamiltonian is assumed to be 
\[
   \hat H_e
     =\hat H
       + \frac{1}{2}\big(
        \frac{1}{m}\hat p_3^2 + m\omega'\mbox{}^2 \hat q_3^2 \big)
         +\frac{1}{2}k_2(\hat q_2-\hat q_3)^2
          +\frac{1}{2}k_3(\hat q_3-\hat q_1)^2,
\]
where $k_j$ denotes the coupling strength between the particles 
and $\omega'$ the resonant frequency of the auxiliary particle; 
these parameters will be appropriately chosen later. 
The particles couple to the common single-mode cavity field that 
can be adiabatically eliminated, which as a result leads to the 
system-bath interaction described by 
$\hat L_e=\sqrt{\kappa}(\hat b_1+\hat b_2+\hat b_3)$ with 
$\hat b_i=(\hat q_i + i\hat p_i)/\sqrt{2}$. 
The system matrices in this case are given by 
\begin{eqnarray*}
& & \hspace*{-1em}
    G=\left( \begin{array}{cc} 
               G_q & O \\
               O & I_3 \\
             \end{array}\right),
\nonumber \\ & & \hspace*{-1em}
    G_q=\left( \begin{array}{ccc} 
               \omega^2+k+k_3 & -k & -k_3 \\
               -k & \omega^2+k+k_2 & -k_2 \\
               -k_3 & -k_2 & \omega'\mbox{}^2 +k_2+k_3 \\
             \end{array}\right),
\end{eqnarray*}
and $c=\sqrt{\kappa/2}(1,1,1,i,i,i)$. 
Here, for simple notation, the matrices are represented in the 
basis of 
$\hat x = 
(\hat q_1, \hat q_2, \hat q_3, \hat p_1, \hat p_2, \hat p_3)^\top$; 
this notation is used throughout this subsection. 
Also, we have assumed $m=1$ without loss of generality.

If $G=O$, the linear coordinate transformation matrix 
$T=(T_1, T_2)$ with 
\begin{eqnarray*}
& & \hspace*{-1em}
    T_1=\left( \begin{array}{cc} 
               \tilde{T}_1 & O \\
               O & \tilde{T}_1 \\
             \end{array}\right),~~
    \tilde{T}_1=\left( \begin{array}{cc} 
                 1/\sqrt{2} & 1/\sqrt{6} \\
                 -1/\sqrt{2} & 1/\sqrt{6} \\
                 0 & -2/\sqrt{6} \\
              \end{array}\right),
\nonumber \\ & & \hspace*{-1em}
    T_2=\left( \begin{array}{cc} 
               \tilde{T}_2 & O \\
               O & \tilde{T}_2 \\
             \end{array}\right),~~
    \tilde{T}_2=\left( \begin{array}{c} 
                 1/\sqrt{3} \\
                 1/\sqrt{3} \\
                 1/\sqrt{3} \\
              \end{array}\right)
\end{eqnarray*}
yields the DF mode $\hat x_{\rm DF}=T_1^\top \hat x$ and the D 
mode $\hat x_{\rm D}=T_2^\top \hat x$. 
Here the goal is to find a condition of $G$ so that the above 
$\hat x_{\rm DF}$ is still a DF mode even when the Hamiltonian 
specified by $G$ is added. 
Now Corollary~\ref{DF mode corollary} is applicable, hence $G$ 
must satisfy $CGT_1=0$, and this yields $\omega'=\omega$. 
In this case, it can be actually verified that $\hat x_{\rm DF}$ 
is the DF mode driven by the Hamiltonian 
$\hat H_{\rm DF}=\hat x_{\rm DF}^\top G_{\rm DF} \hat x_{\rm DF}/2$ 
where
\begin{eqnarray*}
& & \hspace*{-1.8em}
    G_{\rm DF}=T_1GT_1^\top
\nonumber \\ & & \hspace*{-1em}
              =\left( \begin{array}{ccc} 
                 \omega^2+2k+(k_2+k_3)/2 & \sqrt{3}(k_3-k_2)/2 \\
                 \sqrt{3}(k_3-k_2)/2 & \omega^2+3(k_2+k_3)/2 \\
               \end{array}\right) \oplus I_2. 
\end{eqnarray*}

We are particularly interested in the question whether the coupled 
particles with Hamiltonian \eqref{example original} can be 
simulated with the isolated DF mode $\hat x_{\rm DF}=T_1^\top \hat x$. 
Then, because the two particles are identical, the condition 
$2k+(k_2+k_3)/2=3(k_2+k_3)/2$ is necessary. 
By choosing the coupling strength as $k_2=\sqrt{3}k$ and 
$k_3=(2-\sqrt{3})k$, we end up with the expression 
\begin{eqnarray*}
& & \hspace*{-1.9em}
{\small
    G_{\rm DF}
}
\nonumber \\ & & \hspace*{-1.7em}
{\small
       =\left( \begin{array}{ccc} 
           \omega^2+\sqrt{3}k + (3-\sqrt{3})k & -(3-\sqrt{3})k \\
           -(3-\sqrt{3})k & \omega^2+\sqrt{3}k + (3-\sqrt{3})k \\
        \end{array}\right) \oplus I_2. 
}
\end{eqnarray*}
The corresponding Hamiltonian driving the DF mode is thus given by 
\[
    \hat H_{\rm DF}
     =\frac{1}{2}\sum_{j=1,2}\big[
              \hat p_j'\mbox{}^2 
                + (\omega^2+\sqrt{3}k) \hat q_j'\mbox{}^2 \big]
      +\frac{1}{2}(3-\sqrt{3})k(\hat q'_1-\hat q'_2)^2,
\]
where now $\hat x_{\rm DF}=(\hat q'_1, \hat q'_2, \hat p'_1, 
\hat p'_2)^\top$. 
Recall that $m=1$ is assumed. 
We can further prove that $G_{\rm DF}>0$ and also that 
$A_2=T_2^\top AT_2$ is Hurwitz. 
Therefore, from Theorem~\ref{DF stability}, the DF mode 
with the above Hamiltonian is stable. 
Summarizing, we have constructed a stable DF subsystem of two-mode 
coupled particles whose Hamiltonian has the same structure as that 
of the Hamiltonian of the original open system.


\section{Conclusion}

In this paper, we have developed a basic theory of general linear DF 
subsystems, shown some applications, and demonstrated how to actually 
construct a DF subsystem. 
The theory contains a general iff condition for an open linear 
system to have a closed DF subsystem, which fully utilizes the notions 
of controllability and observability in control theory. 
Although, as expected from the literature in the finite dimensional 
case, there should be a number of applications of linear DF subsystems, 
in this paper we have studied only very basic coherent manipulation 
and preservation of the state of a single-mode DF subsystem. 
What was additionally focused on is the idea of synthesis of an open 
linear quantum system having a desirable DF component, to emphasize 
the usefulness of the general results obtained in the former part 
of the paper. 
Several applications of linear DF subsystems theory will be 
presented elsewhere; application to quantum memory is one such 
example \cite{Yamamoto2014}.


\section*{Acknowledgment}

The author thanks M. R. James and I. R. Petersen for their 
insightful comments. 
The author also acknowledges helpful discussions with K. Koga 
and M. Ohnuki.


\section*{Appendix}

This appendix gives some fundamentals of quantum mechanics. 
In quantum mechanics, any physical quantity takes a random value 
when measuring it, hence it is essentially a random variable and 
is called an {\it observable}. 
A biggest difference between quantum observables and classical 
(i.e., non-quantum) random variables is that in general the 
former cannot be measured simultaneously; 
mathematically, this fact is represented by the non-commutativity 
of two observables $\hat A$ and $\hat B$, i.e., 
$[\hat A, \hat B]=\hat A \hat B - \hat B \hat A \neq 0$. 
This implies that an observable should be expressed by a 
matrix (for a finite dimensional system) or an operator (for an 
infinite dimensional system), which is further required to be 
Hermitian for the realization value to be real. 
Accordingly, in quantum mechanics, a probabilistic distribution 
is also represented by a matrix or an operator, and it is called 
a {\it state}. 
A state $\hat \rho$ has to satisfy $\hat \rho=\hat \rho^*\geq 0$ 
and $\Tr\hat \rho=1$. 
When a system of interest is in a state $\hat \rho$ and we measure 
an observable $\hat A$, then the measurement result is randomly 
distributed with mean $\mean{\hat A}=\Tr(\hat A \hat \rho)$.

In an ideal situation, the equation of motion of an observable 
$\hat A(t)$, which is called the {\it Heisenberg equation}, is 
given by $d\hat A(t)/dt=i[\hat H, \hat A(t)]$, where $\hat H$ is 
an observable called {\it Hamiltonian}. 
(For simplicity, $\hat H$ is now assumed to be time-independent.) 
The solution of this equation is 
$\hat A(t) =\hat U^*(t) \hat A(0) \hat U(t)$ with 
$\hat U(t)={\rm e}^{-i\hat H t}$ unitary, hence it is a unitary 
evolution of $\hat A(0)$. 
The mean value of the measurement results of $\hat A(t)$ at time $t$ 
is given by $\mean{\hat A}(t)=\Tr(\hat A(t) \hat \rho)$. 
Note that the same statistics is obtained if we take a picture 
that the observable is fixed to $\hat A=\hat A(0)$ and the state 
evolves in time as $d\hat \rho(t)/dt=-i[\hat H, \hat \rho(t)]$; 
actually we then have $\mean{\hat A}(t)=\Tr(\hat A(t) \hat \rho)
=\Tr(\hat A \hat \rho(t))$. 
This is called the {\it Schr\"{o}dinger picture}, while in the 
{\it Heisenberg picture} we fix the state and deal with an 
observable as a dynamical one.






\end{document}